\def\draft{0}
\documentclass[11pt]{article}
\usepackage{amsmath,amstext,amsthm,amssymb,fullpage}

\usepackage[colorlinks,pagebackref,linkcolor=blue,filecolor = blue,
citecolor = blue, urlcolor  = blue]{hyperref}

\newtheorem{theorem}{Theorem}

\newtheorem{lemma}[theorem]{Lemma}

\newtheorem{definition}[theorem]{Definition}
\newtheorem{claim}[theorem]{Claim}

\makeatletter
\def\fnum@figure{{\bf Figure \thefigure}}
\def\fnum@table{{\bf Table \thetable}}
\long\def\@mycaption#1[#2]#3{\addcontentsline{\csname
  ext@#1\endcsname}{#1}{\protect\numberline{\csname
  the#1\endcsname}{\ignorespaces #2}}\par
  \begingroup
    \@parboxrestore
    \small
    \@makecaption{\csname fnum@#1\endcsname}{\ignorespaces #3}\par
  \endgroup}
\def\mycaption{\refstepcounter\@captype \@dblarg{\@mycaption\@captype}}
\makeatother

\newcommand{\mathify}[1]{\ifmmode{#1}\else\mbox{$#1$}\fi}
\newcommand{\bigO}O

\newcommand{\F}{{\mathbb F}}

\newcommand{\eqdef}{{\stackrel{\rm def}{=}}}

\def\blfootnote{\xdef\@thefnmark{}\@footnotetext}

\usepackage{amsmath,amssymb,amsfonts}
\newcommand{\supp}{{\mathrm {supp}}}

\newcommand{\eps}{\varepsilon}
\renewcommand{\epsilon}{\varepsilon}
\renewcommand{\le}{\leqslant}
\renewcommand{\leq}{\leqslant}
\renewcommand{\ge}{\geqslant}
\renewcommand{\geq}{\geqslant}
\newcommand{\lin}{\mathsf{lin}}

\newcommand{\mynote}[2]{\marginpar{\tiny {\bf {#2}:} \sf {#1}}}

\ifnum\draft=1
\newcommand{\vnote}[1]{\mynote{#1}{VG}}
\newcommand{\snote}[1]{\mynote{#1}{SK}}
\newcommand{\jnote}[1]{\mynote{#1}{JH}}
\else
\newcommand{\vnote}[1]{}
\newcommand{\snote}[1]{}
\newcommand{\jnote}[1]{}
\fi

\title{On the List-Decodability of Random Linear Codes}
\author{{\sc Venkatesan Guruswami}\thanks{Computer Science Department, Carnegie Mellon University. {\tt guruswami@cmu.edu}. Supported in part by NSF CCF 0953155 and a Packard Fellowship.}
\and {\sc Johan H{\aa}stad}\thanks{School of Computer Science and Communication, KTH. {\tt johanh@csc.kth.se}. Research supported by ERC grant 226203.}
\and {\sc Swastik Kopparty}\thanks{CSAIL, MIT. {\tt swastik@mit.edu.} Work was partially done while the author
was an intern at Microsoft Research, New England.}
}
\begin{document}

\maketitle
\thispagestyle{empty}
\begin{abstract}

  For every fixed finite field $\F_q$, $p \in (0,1-1/q)$ and $\varepsilon >
  0$, we prove that with high probability a random subspace $C$ of
  $\F_q^n$ of dimension $(1-H_q(p)-\varepsilon)n$ has the
  property that every Hamming ball of radius $pn$ has at most
  $O(1/\varepsilon)$ codewords.

  This answers a basic open question concerning the list-decodability
  of linear codes, showing that a list size of $O(1/\varepsilon)$
  suffices to have rate within $\varepsilon$ of the ``capacity''
  $1-H_q(p)$. This matches up to constant factors the list-size
  achieved by general random codes, and gives an exponential improvement
  over the best previously known list-size bound of $q^{O(1/\varepsilon)}$.

  The main technical ingredient in our proof is a strong upper bound
  on the probability that $\ell$ random vectors chosen from a Hamming
  ball centered at the origin have too many (more than $\Theta(\ell)$)
  vectors from their linear span also belong to the ball.

\end{abstract}

\newpage

\section{Introduction}

One of the central problems in coding theory is to understand the
trade-off between the redundancy built into codewords (aka the rate of
the code) and the fraction of errors the code enables correcting.
Suppose we are interested in codes over the binary alphabet (for
concreteness) that enable recovery of the correct codeword $c \in
\{0,1\}^n$ from any noisy received word $r$ that differs from $c$ in
at most $pn$ locations.  For each $c$, there are about
${n \choose
  {pn}} \approx 2^{H(p) n}$ such possible received words $r$, where
$H(x) = -x \log_2 x - (1-x) \log_2(1-x)$ stands for the binary entropy
function. 
\vnote{I removed the footnote on approximate equality, as in the proofs we use ${n \choose {pn}} \le 2^{H(p) n}$ which is true without fudging.}
Now for each such $r$, the error-recovery procedure must
identify $c$ as a possible choice for the true codeword. In fact, even if
the errors are randomly distributed and not worst-case, the algorithm
must identify $c$ as a candidate codeword for {most} of these $2^{H(p)
  n}$ received words, if we seek a low decoding error probability.
This implies that there can be at most $\approx
2^{(1-H(p))n}$ codewords, or equivalently the largest rate $R$ of the
code one can hope for is $1-H(p)$.

If we could pack about $2^{(1-H(p))n}$ pairwise disjoint Hamming balls
of radius $p n$ in $\{0,1\}^n$, then one can achieve a rate
approaching $1-H(p)$ while guaranteeing correct and unambiguous
recovery of the codeword from an arbitrary fraction $p$ of
errors. Unfortunately, it is well known that such an asymptotic
``perfect packing'' of Hamming balls in $\{0,1\}^n$ does not exist,
and the largest size of such a packing is at most $2^{(\alpha(p)+o(1))
  n}$ for $\alpha(p) < 1-H(p)$ (in fact $\alpha(p)=0$ for $p \ge 1/4$).
Nevertheless, it turns out that it is possible to pack
$2^{(1-H(p)-\eps)n }$ such Hamming balls such that no $O(1/\eps)$ of
them intersect at a point, for any $\eps > 0$. In fact a random
packing has such a property with high probability.

\noindent {\bf List Decoding.}
This fact implies that it is possible to achieve rate approaching the
optimal $1-H(p)$ bound for correcting a fraction $p$ of {\em
  worst-case} errors in a model called {\em list decoding}. List
decoding, which was introduced independently by Elias and Wonzencraft
in the 1950s~\cite{elias,wozencraft}, is an error-recovery model where
the decoder is allowed to output a small list of candidate codewords
that must include all codewords within Hamming distance $pn$ of the
received word. Note that if at most $pn$ errors occur, the list
decoder's output will include the correct codeword. In addition to the
rate $R$ of the code and the error fraction $p$, list decoding has an
important third parameter, the ``list-size,'' which is the largest
number $L$ of codewords the decoder is allowed to output on any
received word. The list-size thus bounds the maximum ambiguity in the
output of the decoder.

For codes over an alphabet of size $q$, all the above statements hold
with $H(p)$ replaced by $H_q(p)$, where $H_q(x) = x \log_q(q-1) - x
\log_q x - (1-x) \log_q(1-x)$ is the $q$-ary entropy function.

\begin{definition}[Combinatorial list decodability property]
Let $\Sigma$ be a finite alphabet of size $q$, $L \ge 1$ an integer, and $p \in (0,1-1/q)$. 
A code $C \subseteq \Sigma^n$ is said to be $(p,L)$-list-decodable, if for every $x \in \Sigma^n$, there are at most $L$ codewords of $C$ that are at Hamming distance $pn$ or less from $x$. Formally, $|B_n^q(x,p) \cap C| \le L$ for every $x$, where $B^q_n( x, p) \subseteq \Sigma^n$ is the ball of radius $pn$ centered at $x \in \{0,1\}^n$. 
\end{definition}
We restrict $p < 1-1/q$ in the above definition since a random string
differs from each codeword in at most a fraction $1-1/q$ of positions,
and so over alphabet size $q$ decoding from a fraction $1-1/q$ or more
errors is impossible (except for trivial codes).

\noindent {\bf Combinatorics of list decoding.}
A fundamental question in list decoding is to understand the trade-off
between rate, error-fraction, and list-size. For example, what list-size suffices if we want codes of rate within $\eps$ of the optimal $1-H_q(p)$ bound? That is, if we define $L_{q,p}(\eps)$ to be the minimum integer $L$ for which there are $q$-ary $(p,L)$-list-decodable codes of rate at least $1-H_q(p)-\eps$ for infinitely many lengths $n$, how does $L_{q,p}(\eps)$ behave for small $\eps$ (as we keep the alphabet size $q$ and $p \in (0,1-1/q)$ fixed)? 

It is known that unbounded list-size is needed as one approaches the
optimal rate of $1-H_q(p)$. In other words, $L_{q,p}(\eps) \to \infty$
as $\eps \to 0$. This was shown for the binary case in
\cite{blinovsky}, and his result implicitly implies $L_{2,p}(\eps)\ge
\Omega(\log (1/\eps))$ (see \cite{atri-cocoon} for an explicit
derivation of this). For the $q$-ary case, $L_{q,p}(\eps) =
\omega_\eps(1)$ was shown in \cite{blin-q-ary,blin-convexity}.  In the
language of list-decoding, the above-mentioned result on
``almost-disjoint'' sphere packing states that for large enough block lengths, a random code of rate
$1-H_q(p)-\eps$ is $(p,\frac{1}{\eps})$-list-decodable with high
probability. In other words, $L_{q,p}(\eps) \le 1/\eps$. This
result appears in \cite{elias91} (and is based on a previous random
coding argument for linear codes from \cite{ZP81}). The result is
explicitly stated in \cite{elias91} only for $q=2$, but trivially
extends for arbitrary alphabet size $q$.  This result is also tight,
in the sense that with high probability a {\em random} code of rate
$1-H_q(p)-\eps$ is {\em not} $(p,c_{p,q}/\eps)$-list-decodable
w.h.p. for some constant $c_{p,q}>0$~\cite{atri-cocoon}.

An interesting question is to close the exponential gap in the lower and upper bounds on $L_{2,p}(\eps)$, and more generally pin down the asymptotic behavior of $L_{q,p}(\eps)$ for every $q$. The upper bound of $O(1/\eps)$ is perhaps closer to the truth, and it is probably the lower bound that needs strengthening. 

\noindent {\bf Context of this work.}
In this work, we address another fundamental combinatorial question
concerning list-decodable codes, namely the behavior of $L_{q,p}(\eps)$
when restricted to {\em linear codes}. For $q$ a prime power, a
$q$-ary linear code is simply a subspace of $\F_q^n$ ($\F_q$ being the
field of size $q$). 

Most of the well-studied and practically used codes are linear codes.
Linear codes admit a succinct representation in terms of its basis
(called generator matrix). This aids in finding and representing such
codes efficiently, and as a result linear codes are often useful as
``inner'' codes in concatenated code constructions.  

In a linear code,
by translation invariance, the neighborhood of every codeword looks
the same, and this is often a very useful symmetry property. For
instance, this property was recently used in \cite{GS-additive} to
give a black-box conversion of linear list-decodable codes to codes
achieving capacity against a worst-case additive channel (the
linearity of the list-decodable code is crucial for this connection).
Lastly, list-decodability of linear codes brings to the fore some
intriguing questions on the interplay between the geometry
of linear spaces and Hamming balls, and is therefore interesting in
its own right.  For these and several other reasons, it is desirable
to achieve good trade-offs for list decoding via linear codes.

Since linear codes are a highly structured subclass of all codes,
proving the existence of linear codes with list-decodability
properties similar to general codes can be viewed as a strong
``derandomization'' of the random coding argument used to construct
good list-decodable codes. A derandomized family of codes called
``pseudolinear codes'' were put forth in \cite{GI-focs01} since linear
codes were not known to have strong enough list-decoding properties.
Indeed, prior to this work, the results known for linear codes were
substantially weaker than for general codes (we discuss the details
next). {\em Closing this gap is the main motivation behind this
  work.}

\noindent {\bf Status of list-decodability of linear codes.}
Zyablov and Pinsker proved that a random binary linear code of rate
$1-H(p)-\eps$ is $(p,2^{O(1/\eps)})$-list-decodable with high
probability~\cite{ZP81}. The proof extends in a straightforward way to
linear codes over $\F_q$, giving list-size $q^{O(1/\eps)}$ for rate
$1-H_q(p)-\eps$.  Let us define $L^{\lin}_{q,p}(\eps)$ to be the
minimum integer $L$ for which there is an infinite family of
$(p,L)$-list-decodable linear codes over $\F_q$ of rate at least
$1-H_q(p)-\eps$.  The results of \cite{ZP81} thus imply that
$L^{\lin}_{q,p}(\eps) \le \exp(O_q(1/\eps))$.

Note that this bound is {\em exponentially worse} than the $O(1/\eps)$
bound known for general codes. In \cite{elias91}, Elias mentions the
following as the most obvious problem left open left by the random
coding results: {\em Is the requirement of the much larger list size for
linear codes inherent, or can one achieve list-size closer to the
$O(1/\eps)$ bound for general random codes?}

 For the {\em binary} case, the {\em
  existence} of $(p,L)$-list-decodable linear codes of rate at least
$1-H(p)-1/L$ is proven in \cite{GHSZ}. This implies that
$L_{2,p}^{\lin} \le 1/\eps$. There are some results which obtain lower
bounds on the rate for the case of small fixed list-size (at most
$3$)~\cite{blinovsky,blin-book,wei-feng}; these bounds are complicated
and not easily stated, and as noted in \cite{blin2}, are weaker for the
linear case for list-size as small as $5$.

The proof in \cite{GHSZ} is based on a carefully designed potential
function that quantifies list-decodability, and uses the
``semi-random'' method to successively pick good basis vectors for the
code.  The proof only guarantees that such binary linear codes exist
with positive probability, and does not yield a high probability
guarantee for the claimed list-decodability property.
Further, the proof relies crucially on the binary alphabet and
extending it to work for larger alphabets (or even the ternary case)
has resisted all attempts. Thus, for $q > 2$, $L_{q,p}(\eps) \le
\exp(O_q(1/\eps))$ remained the best known upper bound on list-size.
A high probability result for the binary case, and an upper bound of
$L_{q,p}(\eps) \le O(1/\eps)$ for $\F_q$-linear codes, were
conjectured in \cite[Chap. 5]{G-thesis}.

\noindent {\bf Our contribution.} In this work, we resolve the above
open question concerning list-decodability of linear codes over {\em
  all} alphabets. In particular, we prove that $L^{\lin}_{q,p}(\eps)
\le C_{q,p}/\eps$ for a constant $C_{q,p} < \infty$. Up to constant factors,
this matches the best known result for general, non-linear
codes. Further, our result in fact shows that a random $\F_q$-linear code of
rate $1-H_q(p)-\eps$ is $(p,C_{p,q}/\eps))$-list-decodable {\em with
  high probability}. This was not known
even for the case $q=2$. The high probability claim implies an
efficient randomized Monte Carlo construction of such list-decodable
codes.

We now briefly explain the difficulty in obtaining good bounds for
list-decoding linear codes and how we circumvent it. This is just a
high level description; see the next section for a more technical
description of our proof method.

Let us recall the straightforward random coding method that shows the
list-decodability of random (binary) codes. We pick a code $C
\subseteq \{0,1\}^n$ by uniformly and independently picking $M=2^{Rn}$
codewords. To prove it is $(p,L)$-list-decodable, we fix a center $y$
and a subset $S$ of $(L+1)$ codewords of $C$. Since these codewords are
independent, the probability that all of them land in the ball of radius
$pn$ around $y$ is at most $\bigl(\frac{2^{H(p)
    n}}{2^n}\bigr)^{L+1}$. A union bound over all $2^n$ choices of $y$
and at most $M^{L+1}$ choices of $S$ shows that if $R \le 1-H(p) - 1/L$,
the code fails to be $(p,L)$-list-decodable with probability at most
$2^{-\Omega(n)}$.

Attempting a similar argument in the case of random linear codes,
defined by a random linear map $A : \F_2^{Rn} \rightarrow \F_2^n$,
faces several immediate obstacles.  The $2^{Rn}$ codewords of a random
linear code are not independent of one another; in fact the points of
such a code are highly correlated and not even $3$-wise independent
(as $A(x+y) = Ax+Ay$). However, any $(L+1)$ distinct codewords
$Ax_1,Ax_2,\dots,Ax_{L+1}$ must contain a subset of $\ell \ge
\log_2(L+1)$ independent codewords, corresponding to a subset
$\{x_{i_1}, \dots, x_{i_\ell}\}$ of {\em linearly independent}
message vectors. This lets one mimic the argument for the random code case
with $\log_2(L+1)$ playing the role of $L+1$. However, as a result, it
leads to the exponentially worse list-size bounds.

To get a better result, we somehow need to control the ``damage''
caused by subsets of codewords of low rank. This is the crux of our
new proof. Stated loosely and somewhat imprecisely, we prove a strong
upper bound on the fraction of such low rank subsets, by proving that
if we pick $\ell$ random vectors from the Hamming ball $B_n(0,p)$ (for
some constant $\ell$ related to our target list-size $L$), it is
rather unlikely that more than $\Theta(\ell)$ of the $2^\ell$ vectors
in their span will also belong to the ball $B_n(0,p)$.  (See
Theorem~\ref{thm:cbound} for the precise statement.) This ``limited
correlation'' between linear subspaces and Hamming balls is the main
technical ingredient in our proof. It seems like a basic and powerful
probabilistic fact that might find other applications. The argument
also extends to linear codes over $\F_q$ after some adaptations.

\section{Results and Methods}

Our main result is that random linear codes in $\F_2^n$ of rate $1 - H(p)- \epsilon$ can be list-decoded
from $p$-fraction errors with list-size only $O(\frac{1}{\epsilon})$. We also show the analogous result
for random $q$-ary linear codes.

\begin{theorem}
\label{thm:main}
Let $p \in (0, 1/2)$. Then there exist constants $C_p, \delta > 0$, such that
for all $\epsilon > 0$ and all large enough integers $n$, letting $R = 1 - H(p) - \epsilon$, if $\mathcal C \subseteq \F_2^n$
is a random linear code of rate $R$, then
$$\Pr[\mathcal C \mbox{ is $(p, \frac{C_p}{\epsilon})$-list-decodable}] > 1 - 2^{-\delta n}.$$
\end{theorem}

The proof begins by simplifying the problem to its combinatorial
core. Specifically, we reduce the problem of studying the {\em
  list-decodability} of a random linear code of {\em linear} dimension
to the problem of studying the {\em weight-distribution} of certain
random linear codes of {\em constant} dimension.  The next theorem
analyzes the weight distribution of these constant-dimensional random
linear codes. The notation $B_n(x,p)$ refers to the Hamming ball of radius $pn$ centered at $x \in \F_2^n$.

\begin{theorem}[Span of random points in $B_n(0,p)$]
\label{thm:cbound}
For every $p \in ( 0, 1/2)$, there is a constant $C > 0$, such that
for all $n$ large enough and all $\ell = o(\sqrt{n})$, if $X_1,
\ldots, X_\ell$ are picked independently and uniformly at random from
$B_n(0, p)$, then
$$ \Pr [ | \mathrm{span}(\{X_1, \ldots, X_\ell\}) \cap B_n(0, p) | > C \cdot\ell] \leq 2^{-5n}.$$
\end{theorem}

We now give a brief sketch of the proof of Theorem~\ref{thm:cbound}.
Index the elements of $\mathrm{span}(\{X_1, \ldots, X_\ell\})$ as
follows: for $v \in \F_2^{\ell}$, let $X_v$ denote the random vector
$\sum_{i=1}^{\ell} v_i X_i$.  Fix an arbitrary $S \subseteq \F_2^\ell$
of cardinality $C \cdot\ell$, and let us study the event $E_S$: that all
the vectors $(X_v)_{v \in S}$ lie in $B_n(0,p)$. If none of the events
$E_S$ occur, we know that $| \mathrm{span}(\{X_1, \ldots, X_\ell\})
\cap B_n(0, p) | \leq C\cdot\ell.$

The key technical step is a Ramsey-theoretic lemma (Lemma~\ref{lem:cinc}, stated below) which says
that large sets $S$ automatically have the property that some
translate of $S$ contains a certain structured subset 
(which we call an ``increasing chain'').
This structured subset allows us to give strong upper bounds on the probability that all the vectors
$(X_v)_{v \in S}$ lie in $B_n(0,p)$.
Applying this to each $S \subseteq \F_2^\ell$ of cardinality $C \ell$ and taking a union bound gives
Theorem~\ref{thm:cbound}.

To state the Ramsey-theoretic lemma (Lemma~\ref{lem:cinc}), we first define increasing chains. For a vector $v \in \F_2^{\ell}$, the {\em support} of $v$, denoted $\supp(v)$, is defined to be the set of its nonzero coordinates.

\begin{definition}
A sequence of vectors $v_1, \ldots, v_d \in \F_2^{\ell}$ is called an $c$-increasing chain of length $d$,
if for all $j \in [d]$, 
$$\left|\supp(v_j) \setminus \left(\bigcup_{i = 1}^{j-1}\supp(v_i)\right)\right| \geq c.$$
\end{definition}

We now state the Ramsey-theoretic lemma that plays the central role in Theorem~\ref{thm:cbound}.
The proof appears in Section~\ref{sec:cinc}, where it is proved using the Sauer-Shelah lemma.

\begin{lemma}
\label{lem:cinc}
For all positive integers $c,\ell$ and $L \le 2^\ell$, the following holds.
For every $S \subseteq
\F_2^\ell$ with $|S| = L$, there is a $w \in \F_2^\ell$ such that $S+ w$
has an $c$-increasing chain of length at least $\frac{1}{c}(\log
\frac{L}{2}) - (1 - \frac{1}{c})(\log \ell) $.
\end{lemma}

\subsection{Larger alphabet}

Due to their geometric nature, our arguments generalize to the case of $q$-ary alphabet
(for arbitrary constant $q$) quite easily. Below we state our main theorem for the case of $q$-ary alphabet.

\begin{theorem}
\label{thm:qmain}
Let $q$ be a prime power and let $p \in (0, 1-1/q)$. Then there exist constants $C_{p,q}, \delta > 0$, such that
for all $\epsilon > 0$, letting $R = 1 - H_q(p) - \epsilon$, if $\mathcal C \subseteq \F_q^n$
is a random linear code of rate $R$, then
$$\Pr[\mathcal C \mbox{ is $(p, \frac{C_{p,q}}{\epsilon})$-list-decodable}] > 1 - 2^{-\delta n}.$$
\end{theorem}

The proof of Theorem~\ref{thm:qmain} has the same basic outline as the
proof of Theorem~\ref{thm:main}.  In particular, it proceeds via a
$q$-ary analog of Theorem~\ref{thm:cbound}.  The only notable
deviation occurs in the proof of the $q$-ary analog of
Lemma~\ref{lem:cinc}.  The traditional generalization of the
Sauer-Shelah lemma to larger alphabets turns out to be unsuitable for
this purpose. Instead, we formulate and prove a non-standard
generalization of the Sauer-Shelah lemma for the larger alphabet case
which is more appropriate for this situation. Details appear in
Section~\ref{sec:qary}.


\section{Proof of Theorem~\ref{thm:main}}
\label{sec:ctogeneral}

Let us start by restating our main theorem.

\noindent {\bf Theorem~\ref{thm:main} (restated)}\ \ 
{\it
Let $p \in (0, 1/2)$. Then there exist constants $C_p, \delta > 0$, such that
for all $\epsilon > 0$ and all large enough integers $n$, letting $R = 1 - H(p) - \epsilon$, if $\mathcal C \subseteq \F_2^n$
is a random linear code of rate $R$, then
$$\Pr[\mathcal C \mbox{ is $(p, \frac{C_p}{\epsilon})$-list-decodable}] > 1 - 2^{-\delta n}.$$
}

\begin{proof}
Pick $C_p = 4 C$, where $C$ is the constant from Theorem~\ref{thm:cbound}. Pick $\delta = 1$. Take $L = \frac{C_p}{\epsilon}$.

Let $\mathcal C$ be a random $Rn$ dimensional linear subspace of $\F_2^n$.
We want to show that
\begin{equation}
\label{eq:bound1}
\Pr_{\mathcal C} [\exists x \in \F_2^n \mbox{ s.t. } |B_n(x, p) \cap \mathcal C | > L] < 2^{-\delta n}.
\end{equation}

Let $x \in \F_2^n$ be picked uniformly at random. We will work towards Equation~\eqref{eq:bound1} by
studying the following quantity.
$$ \Delta \eqdef \Pr_{\mathcal C, x}[ |B_n(x, p) \cap \mathcal C | > L ].$$
Note that to prove Equation~\eqref{eq:bound1}, it suffices to show
that\footnote{We could even replace the $2^{-n}$ by
  $2^{-(1-R)n}$. Indeed, for every $\mathcal C$ for which there is a
  ``bad'' $x$, we know that there are $2^{Rn}$ ``bad'' $x$'s (the
  translates of $x$ by $\mathcal C$).}
$$ \Delta < 2^{-\delta n} \cdot 2^{-n}.$$
\vnote{We can also replace $2^{-n}$ by $2^{(H(p)-1)n}$. This would better than $2^{(R-1)n}$ for low rates or $p \to 1/2$. A commented footnote is in the latex file.}


Now for each $\ell \in [\log(L+1), L+1]$, let 
$\mathcal F_\ell$ be the set of all $(v_1, \ldots, v_{\ell}) \in B_n(0,p)^\ell$ such
that $v_1, \ldots, v_\ell$ are linearly independent and 
$|\mathrm{span}(v_1, \ldots, v_{\ell}) \cap B_n(0,p)^{\ell}| > L$.
Let $\mathcal F = \bigcup_{\ell = \log (L+1)}^{L+1} \mathcal F_\ell$ 

For each $\mathbf v = (v_1, \ldots, v_\ell) \in \mathcal F$, let $\{ \mathbf v \}$
denote the set $\{ v_1, \ldots, v_{\ell}\}$.

We now bound $\Delta$. Notice that if $|B_n(x,p) \cap \mathcal C| >
L$, then there must be some $\mathbf v \in \mathcal F$ for which
$B_n(x,p) \cap \mathcal C \supseteq x + \{ \mathbf v \}$.  Indeed, we
can simply take $\mathbf v$ to be a maximal linearly independent
subset of $(B_n(x,p) \cap \mathcal C) + x$ if this set has size at
most $L+1$, and any linearly independent subset of $(B_n(x,p) \cap
\mathcal C) + x$ of size $L+1$ otherwise.

Therefore, by the union bound,
\begin{align}
\Delta &\leq \sum_{\mathbf v \in \mathcal F} \Pr_{\mathcal C, x} [ B_n(x,p) \cap \mathcal C \supseteq x + \{ \mathbf v \}]\\
&= \sum_{\mathbf v \in \mathcal F} \Pr_{\mathcal C, x} [B_n(0, p) \cap (\mathcal C + x) \supseteq \{ \mathbf v \}]\\
&\leq \sum_{\mathbf v \in \mathcal F} \Pr_{\mathcal C, x} [B_n(0, p) \cap (\mathcal C + \{0, x\} ) \supseteq \{ \mathbf v \}]\\
&= \sum_{\mathbf v \in \mathcal F} \Pr_{\mathcal C^*} [B_n(0,p) \cap \mathcal C^* \supseteq \{ \mathbf v \}],
\label{eq:b2}
\end{align}
where $\mathcal C^*$ is the code $\mathcal C + \{0, x\}$ which is a
random $Rn + 1$ dimensional subspace.

The last probability can be bounded as follows.  By the linear
independence of $v_1, \ldots, v_\ell$, the probability that $v_j \in
\mathcal C^*$ conditioned on $\{v_1, \ldots, v_{j-1}\} \subseteq
\mathcal C^*$ is precisely the probability that a given point in a
$n+1-j$ dimensional space lies in a $Rn+1-j$ dimensional subspace, and
hence this conditional probability is exactly $2^{Rn+1-n}$.  We can
hence conclude that
\begin{equation}
\label{eq:b3}
\Pr_{\mathcal C^*}[ \mathcal C^*\supseteq \{ \mathbf v \} ] = \left(\frac{ 2^{Rn + 1} }{2^n} \right)^\ell.
\end{equation}

Putting together Equations~\eqref{eq:b2} and \eqref{eq:b3}, we have
\begin{align*}
\Delta &\leq \sum_{\mathbf v \in \mathcal F} \Pr_{\mathcal C^*} [B_n(0,p) \cap \mathcal C^* \supseteq \{ \mathbf v \}]
\leq \sum_{\ell = \log (L+1)}^{L+1} \sum_{\mathbf v \in \mathcal F_\ell} \Pr_{\mathcal C^*} [\mathcal C^* \supseteq \{ \mathbf v \}]\\
&\leq \sum_{\ell = \log (L+1)}^{L+1} \sum_{\mathbf v \in \mathcal F_\ell} \left(\frac{ 2^{Rn + 1} }{2^n} \right)^\ell 
\leq \sum_{\ell = \log (L+1)}^{L+1} |\mathcal F_{\ell}| \cdot \left(\frac{ 2^{Rn + 1} }{2^n} \right)^\ell \\
\end{align*}

We now obtain an upper bound on $|\mathcal F_\ell|$.
We have two cases depending on the size of $\ell$.
\begin{itemize}
\item {\bf Case 1:} $\ell < 4/\epsilon$. In this
case, we notice that $\frac{|\mathcal F_{\ell}|}{|B_n(0,p)|^\ell}$ is a lower bound on the probability that $\ell$ points $X_1, \ldots, X_\ell$ chosen uniformly at random from $B_n(0,p)$
have $|\mathrm{span}(\{X_1, \ldots, X_\ell\}) \cap B_n(0,p)| > L.$
Since $L > C \cdot \ell$, Theorem~\ref{thm:cbound} tells us
that this probability is bounded from above by $2^{-5n}$. Thus, in this case $|\mathcal F_{\ell}| \leq |B_n(0,p)|^{\ell} 2^{-5n} \leq 2^{n \ell H(p)} \cdot 2^{-5n}.$
\item {\bf Case 2:} $\ell \geq 4/\epsilon$. In this case, we have the trivial bound of $|\mathcal F_\ell| \leq |B_n(0,p)|^{\ell} \leq 2^{n\ell H(p)}$.
\end{itemize}
\noindent
Thus, we may bound $\Delta$ by:
\begin{align*}
\Delta & \leq \sum_{\ell=\log L}^{\lfloor 4/\epsilon\rfloor} |\mathcal F_{\ell}| \cdot \left(\frac{ 2^{Rn+1} }{2^n} \right)^\ell +  \sum_{\ell=\lceil 4/\epsilon\rceil}^{L} |\mathcal F_{\ell}| \cdot \left(\frac{ 2^{Rn+1} }{2^n} \right)^\ell\\
&\leq \sum_{\ell = \log L}^{\lfloor 4/\epsilon\rfloor} 2^{n\ell H(p)} 2^{-5n} \left(\frac{ 2^{Rn+1} }{2^n} \right)^\ell +  \sum_{\ell=\lceil 4/\epsilon\rceil}^{L} 2^{n\ell H(p)}  \left(\frac{ 2^{Rn+1} }{2^n} \right)^{\ell}\\
& \leq 2^{-5n} \cdot 4/ \epsilon + L \cdot 2^{-(\epsilon n)\cdot (4/\epsilon) }\\
& \leq 2^{-\delta n} \cdot 2^{-n}
\end{align*}
as desired.
\end{proof}

\section{Proof of Theorem~\ref{thm:cbound}}
\label{sec:cbound}
In this section, we prove Theorem~\ref{thm:cbound} which bounds the probability that the span of $\ell$ random points
in $B_n(0,p)$ intersects $B_n(0,p)$ in more than $C \cdot \ell$ points, for some large constant $C$.
We use the following simple fact.
\begin{lemma}
\label{lem:deltap}
For every $p \in (0, 1/2)$, there is a $\delta_p > 0$ such that for
all large enough integers $n$ and every $x \in \F_2^n$, the
probability that two uniform independent samples $w_1, w_2$ from
$B_n(0,p)$ are such that $w_1 + w_2 \in B_n(x,p)$ is at most
$2^{-\delta_p n}$.
\end{lemma}
\noindent
{\bf Sketch of proof.} The point $w_1+w_2$ is essentially
a random point in $B_n(0,2p-2p^2)$.  The probability that it lies
in the smaller ball $B_n(x,p)$ is easily seen to be maximal
when $x=0$ and is then exponentially small.\qed

\noindent
{\bf Theorem~\ref{thm:cbound} (restated)}\ \ 
{\it
For every $p \in ( 0, 1/2)$, there is a constant $C > 0$, such that for all $n$ large enough and all
$\ell = o(\sqrt{n})$, if
$X_1, \ldots, X_\ell$ are picked independently and 
uniformly at random from $B_n(0, p)$,
then
$$ \Pr [ | \mathrm{span}(\{X_1, \ldots, X_\ell\}) \cap B_n(0, p) | > C\cdot\ell] \leq 2^{-5n}.$$
}
\begin{proof}
  Set $L = C \cdot \ell$ and let $c = 2$. Let $\delta_p > 0$ be the
  constant given by Lemma~\ref{lem:deltap}. Let
\[ d = \bigg\lfloor \frac{1}{c}\log \frac{L}{2} - \Bigl(1 - \frac{1}{c}\Bigr)\log\ell \bigg\rfloor \ge \frac{1}{2} \log \frac{L}{2\ell} - 1 = \frac{1}{2}
\log \frac{C}{8} \ . \] 
\noindent For a vector $u \in \F_2^{\ell}$, let $X_u$
denote the random variable $\sum_{i} u_i X_i$.

We begin with a claim which bounds the probability of a particular
collection of linear combinations of the $X_i$ all lying within
$B_n(0,p)$. At the heart of this claim lies the Ramsey-theoretic
Lemma~\ref{lem:cinc}.

\begin{claim}
\label{claim:perS}
For each $S \subseteq \F_2^\ell$ with $|S| =L+1$, 
\begin{equation}
\label{eq:perS}
\Pr[ \forall v \in S, X_v \in B_n(0,p)] < 2^n \cdot 2^{-\delta_p dn}.
\end{equation}
\end{claim}
\begin{proof}
  Let $w$ and $v_1, \ldots, v_d \in S$ be as given by
  Lemma~\ref{lem:cinc}. That is, $v_1+w, v_2+w, \cdots, v_d+w$ is an
  $c$-increasing sequence. 
Then,
\begin{align}
\Pr [ \forall v \in S, X_v \in B_n(0,p)] &\leq  \Pr [ \forall j \in [d], X_{v_j} \in B_n(0,p)] \\
&= \Pr [ \forall j \in [d], X_{v_j} +  X_w \in B_n(X_w,p)]\\
&= \Pr [ \forall j \in [d], X_{v_j + w} \in B_n(X_w,p)]
\label{eq:bnxw}
\end{align}
We now bound the probability that there exists $y \in \F_2^n$
such that for all $j \in [d]$, $X_{v_j + w} \in B_n(y, p)$.
Fix $y \in \F_2^n$. We have:
\begin{align}
\Pr [ \forall j \in [d], X_{v_j +w} \in B_n(y,p)] &\leq \prod_{j = 1}^d \Pr\biggl[ X_{v_j + w} \in B_n(y,p) \mid ( X_t : t \in \Bigl(\bigcup_{i = 1}^{j-1}\supp(v_i + w)\Bigr)\biggr]\\
&\leq \bigl( 2^{-\delta_p  n}\bigr)^d.
\label{eq:bnxw1}
\end{align}
The last inequality follows from applying Lemma~\ref{lem:deltap} with
$w_1$ and $w_2$ being vectors $X_{i_1}$ and $X_{i_2}$, where $i_1, i_2$ are
two distinct elements of $\supp(v_j + w) \setminus \left(\bigcup_{i = 1}^{j-1}\supp(v_i +
  w)\right)$, and $x = y + \sum_{k \in [\ell], k\not\in \{i_1, i_2\}}
(v_j+w)_k X_k$.
Taking a union bound of Equation~\eqref{eq:bnxw1} over all $y \in \F_2^n$, we see that
$$ \Pr[ \exists y \in F_2^n \mbox{ s.t. } \forall j \in [d], X_{v_j + w} \in B_n(y,p)] \leq 2^n \cdot 2^{-\delta_p n d}.$$
Combining this with Equation~\eqref{eq:bnxw} completes the proof of the claim.
\end{proof}

Given this claim, we now bound the probability that more than $L$
elements of ${\mathrm {span}}(\{X_1, \ldots, X_\ell\})$ lie inside
$B_n(0,p)$.  This event occurs if and only if for some set $S
\subseteq \F_2^\ell$ with $|S| = L+1$, it is the case that $\forall v
\in S$, $X_v \in B_n(0,p)$.  Taking a union bound of~\eqref{eq:perS}
over all such $S$, we see that the probability that there exists some
$S \subseteq \F_2^\ell$ with $|S| = L+1$ such that $\forall v \in S,
X_v \in B_n(0,p)$ is at most $2^{\ell (L+1)} \cdot 2^{n} \cdot
2^{-\delta_p d n}$.  Taking $C$ to be a large enough constant so that
$d \ge \frac{1}{2} \log\frac{C}{8} > \frac{12}{\delta_p}$, the theorem
follows.
\end{proof}

\section{Proof of Lemma~\ref{lem:cinc}}
\label{sec:cinc}

In this section, we will prove Lemma~\ref{lem:cinc}, which finds a
large $c$-increasing chain in some translate of any large enough set
$S \subseteq \F_2^\ell$.

We will use the Sauer-Shelah Lemma.

\begin{lemma}[Sauer-Shelah~\cite{sauer, shelah}]
For all integers $\ell, c$, and for any set $S \subseteq \{0,1\}^\ell$, if $|S| > 2 \ell^{c-1}$, then
there exists some set of coordinates $U \subseteq [\ell]$ with $|U| = c$ such that 
$\{ v|_U \mid v \in S \} = \{0, 1\}^U$.
\end{lemma}

\noindent
{\bf Lemma~\ref{lem:cinc} (restated)}\ \ 
{\it
For all positive integers $c,\ell$ and $L \le 2^\ell$, the following holds.
For every $S \subseteq
\F_2^\ell$ with $|S| = L$, there is a $w \in \F_2^\ell$ such that $S+ w$
has an $c$-increasing chain of length at least $\frac{1}{c}(\log
\frac{L}{2}) - (1 - \frac{1}{c})(\log \ell) $.
}

\begin{proof}
We prove this by induction on $\ell$. The claim holds trivially for $\ell \le c$, so assume $\ell > c$. 

If $L \leq 2 \ell^{c-1}$, then again the lemma holds trivially.
Otherwise, by the Sauer-Shelah lemma, we get a set $U$ of $c$
coordinates such that for each $u \in \F_2^U$, there is some $v \in S$
such that $v|_U = u$.  We will represent elements of $\F_2^{\ell}$ in
the form $(u, v')$ where $u \in \F_2^U$ and $v' \in
\F_2^{[\ell]\setminus U}$.

Let $u_0 \in \F_2^U$ be a vector such that $|\{ v \in S \mid v|_{U} =
u_0\}|$ is at least $L / 2^c$ (we know that such a $u$ exists by
averaging).  Let $S' \subseteq \F_2^{[\ell]\setminus U}$ be given by
$S' = \{ v|_{[\ell]\setminus U} \mid v|_{U} = u\}$.  By choice of $u$,
we have $|S'| \geq L/2^c$.


By the induction hypothesis, there exist $w' \in \F_2^{\ell-c}$ and $v'_1,
\ldots, v'_{d'} \in S'$ such that for each $j \in [d']$,
$$\Bigl|\supp(v'_j + w') \setminus \Bigl(\bigcup_{i = 1}^{j-1}\supp(v'_i + w')\Bigr)\Bigr| \geq c.$$
for $d'  \ge \frac{1}{c} \log (L/2^{c+1}) - ( 1- \frac{1}{c}) \log (\ell - c)$.

Let $d=d'+1$. Note that $d
\ge \frac{1}{c} \log (L/2) - ( 1- \frac{1}{c}) \log (\ell - c) \ge
\frac{1}{c} \log (L/2) - ( 1- \frac{1}{c}) \log \ell$.  For $i \in
[d']$, let $v_i = (u_0, v'_i) \in \F_2^{\ell}$. Let $v_d$ be any
vector in $S$ with $(v_d)|_U = \neg u_0$, the bitwise complement of
$u_0$. Let $w = (u_0, w')$.  We claim that $w$ and $v_1, \ldots, v_d$
satisfy the desired properties.

\noindent
Indeed, for each $j \in [d']$, we have
\begin{align*}
\biggl|\supp(v_j + w) \setminus \Bigl(\bigcup_{i = 1}^{j-1}\supp(v_i + w)\Bigr)\biggr| &=
\biggl|\supp(v'_j + w') \setminus \Bigl(\bigcup_{i = 1}^{j-1}\supp(v'_i + w')\Bigr)\biggr| \geq c.
\end{align*}
\noindent
Also
$$\left|\supp(v_d + w) \setminus \left(\bigcup_{i = 1}^{d-1}\supp(v_i + w)\right)\right| \geq \left| \supp(v_d + w) \setminus ([\ell] \setminus U )\right| = |U| = c.$$
\noindent
Thus for all $j \in [d]$, we have $\Bigl|\supp(v_j + w) \setminus \bigl(\bigcup_{i = 1}^{j-1}\supp(v_i + w)\bigr)\Bigr| \geq c$,
as desired.
\end{proof}

\section{Larger alphabets}
\label{sec:qary}

As mentioned in the introduction the case of $q$-ary alphabet is nearly identical to the
case of binary alphabet. We only highlight the differences.
As before, the crux turns out to be the problem of studying the weight distribution of certain random constant-dimensional codes.

\begin{theorem}[$q$-ary span of random points in $B^q_n(0,p)$]
\label{thm:qcbound}
For every prime-power $q$ and every $p \in ( 0, 1-1/q)$, there is a constant $C_q > 0$, such that for all $n$ large enough and all
$\ell = o(\sqrt{n})$, if
$X_1, \ldots, X_\ell$ are picked independently and 
uniformly at random from $B^q_n(0, p)$,
then
$$ \Pr [ | \mathrm{span}(\{X_1, \ldots, X_\ell\}) \cap B^q_n(0, p) | > C_q \cdot\ell] \leq q^{-5n}.$$
\end{theorem}

The proof of Theorem~\ref{thm:qcbound} proceeds as before, by bounding the probability via a large $c$-increasing
chain. The $c$-increasing chain itself is found in an analog of Lemma~\ref{lem:cinc} for $q$-ary alphabet. We first need a definition.

\begin{definition}
A sequence of vectors $v_1, \ldots, v_d \in [q]^{\ell}$ is called an $c$-increasing chain of length $d$,
if for all $j \in [d]$, 
$$\left|\supp(v_j) \setminus \left(\bigcup_{i = 1}^{j-1}\supp(v_i)\right)\right| \geq c.$$
\end{definition}

Now we have the following lemma.

\begin{lemma}[$q$-ary increasing chains Ramsey]
\label{lem:qcinc}
For every prime power $q$, and all positive integers $c,\ell$ and $L \le q^\ell$, the following holds. For every $S \subseteq \F_q^\ell$ with $|S| = L$, there is a
$w \in \F_q^\ell$ such that
$S+ w$ has an $c$-increasing chain of length at least $\frac{1}{c}\log_q \bigl(\frac{L}{2}\bigr) - (1 - \frac{1}{c})\log_q ((q-1)\ell)$.
\end{lemma}

The proof of Lemma~\ref{lem:qcinc} needs a non-standard generalization of the Sauer-Shelah lemma to larger alphabet described in the next section.

\subsection{A $q$-ary Sauer-Shelah lemma}

The traditional generalization of the Sauer-Shelah lemma to large
alphabets is the Karpovsky-Milman lemma~\cite{KMqary}, which roughly states that
given $S \subseteq [q]^\ell$ of cardinality at least $(q-1)^l
l^{c-1}$, there is a set $U$ of $c$ coordinates such that for every $u
\in [q]^U$, there is some $v \in S$ such that the restriction $v|_U$
equals $u$.  Applying this lemma in our context, once $q > 2$,
requires us to have a set $S > 2^\ell$, which turns out to lead to
exponential list size bounds. Fortunately, the actual property needed
for us is slightly different. We want a bound $B$ (ideally polynomial
in $\ell$) such that for any set $S \subseteq [q]^\ell$ of cardinality
at least $B$, there is a set $U$ of $c$ coordinates such that for
every $u \in [q]^U$, there is some $v \in S$ such that the restriction
$v|_U$ {\em differs from $u$ in every coordinate of $U$}.  It turns
out that this weakened requirement admits polynomial-sized $B$.

We state and prove this generalization of the Sauer-Shelah lemma below.

\begin{lemma}[$q$-ary Sauer-Shelah]
\label{lem:qsauer}
For all integers $q, \ell, c$, for any set $S \subseteq [q]^\ell$, if $|S| > 2 \cdot ((q-1) \cdot \ell)^{c-1}$, then
there exists some set of coordinates $U \subseteq [\ell]$ with $|U| = c$ such that for every $u \in [q]^U$,
there exists some $v \in S$ such that $u$ and $v|_{U}$ differ in every coordinate.
\end{lemma}
\begin{proof}
We prove this by induction on $\ell$ and $c$.
If $c= 1$, then $|S| > 2$ and the result holds by letting $U$ equal any coordinate on which
not all elements of $S$ agree.

Now assume $c > 1$.
Represent an element $x$ of $[q]^{\ell}$ as a pair $(y, b)$, where $y \in [q]^{\ell-1}$ consists of the first
$\ell-1$ coordinates of $x$ and $b \in [q]$ is the last coordinate of $x$.

Consider the following subsets of $[q]^{\ell-1}$.
$$S_1  = \{ y \in [q]^{\ell-1} \mid \mbox{ for at least 1 value of $b \in [q]$, $(y, b) \in S$}\}.$$
$$S_2  = \{ y \in [q]^{\ell-1} \mid \mbox{ for at least 2 values of $b \in [q]$, $(y, b) \in S$}\}.$$

\noindent
Note that $|S| \leq  (|S_1|-|S_2|) + q |S_2| = |S_1| + (q-1) |S_2|$.
By assumption,  
$$|S| > 2 \cdot ((q-1) \cdot \ell)^{c-1} \geq 2 \cdot ((q-1) \cdot (\ell-1))^{c-1} + (q-1) \left( 2 \cdot ((q-1) \cdot (\ell-1))^{c-2}\right),$$
(using the elementary inequality $\ell^{c-1} \geq (\ell-1)^{c-1} +
(\ell-1)^{c-2}$).  Thus, either $|S_1| > 2\cdot ((q-1) \cdot
(\ell-1))^{c-1}$, or else $|S_2| > 2 \cdot ((q-1) \cdot
(\ell-1))^{c-2}$.

\noindent
We now prove the desired claim in each of these cases.

\smallskip
\noindent {\bf Case 1:} $|S_1| > 2\cdot ((q-1) \cdot (\ell-1))^{c-1}$. In this case,
we can apply the induction hypothesis to $S_1$ with parameters $\ell-1$ and $c$, and get a 
subset of $U$ of $[\ell-1]$ of cardinality $c$. Then the set $U$ has the desired property.

\smallskip
\noindent {\bf Case 2:} $|S_2| > 2 \cdot ((q-1) \cdot
(\ell-1))^{c-2}$. In this case, we apply the induction hypothesis to
$S_2$ with parameters $\ell-1$ and $c-1$, and get a subset $U$ of
$[\ell -1]$ of cardinality $c-1$. Then the set $U \cup \{ \ell\}$ has
the desired property. Indeed, take any vector $u \in [q]^{U \cup
  \{\ell\}}$.  Let $u' = u |_U$. By the induction hypothesis, we know
that there is a $v \in S_2$ such that $v|_U$ differs from $u'$ in
every coordinate of $U$. Now we know that there are at least two $b
\in [q]$ such that $(v, b) \in S$. At least one of these $b$ will be
such that $(v, b)$ differs from $u$ in every coordinate of $U \cup
\{\ell\}$, as desired.
\end{proof}

In the next section, we use the above lemma to prove the
Ramsey-theoretic $q$-ary increasing chain claim
(Lemma~\ref{lem:qcinc}).


\section{Proof of $q$-ary increasing chain lemma}
\label{app:q-ary-proof}

In this section, we prove Lemma~\ref{lem:qcinc}, which we restate
below for convenience.

\smallskip
\noindent {\bf Lemma~\ref{lem:qcinc} (restated)}\ \ {\it For every
  prime power $q$, and all positive integers $c,\ell$ and $L \le
  q^\ell$, the following holds. For every $S \subseteq \F_q^\ell$ with
  $|S| = L$, there is a $w \in \F_q^\ell$ such that $S+ w$ has an
  $c$-increasing chain of length at least $\frac{1}{c}\log_q
  \bigl(\frac{L}{2}\bigr) - (1 - \frac{1}{c})\log_q ((q-1)\ell)$.  }

\begin{proof}
We prove this by induction on $\ell$. The claim holds trivially for $\ell \le c$, so assume $\ell > c$. 

If $L \leq 2 ((q-1) \cdot \ell)^{c-1}$, then again the lemma holds
trivially.  Otherwise, by Lemma~\ref{lem:qsauer} we get a set $U$ of
$c$ coordinates such that for each $u \in \F_q^U$, there is some $v
\in S$ such that $v|_U$ differs from $u$ in every coordinate.  We will
represent elements of $\F_q^{\ell}$ in the form $(u, v')$ where $u \in
\F_q^U$ and $v' \in \F_q^{[\ell]\setminus U}$.

Let $u_0 \in \F_q^U$ be a vector such that $|\{ v \in S \mid v|_{U} = u_0\}|$ is
at least $L / q^c$ (we know that such a $u$ exists by averaging).
Let $S' \subseteq \F_q^{[\ell]\setminus U}$ be given by
$S' = \{ v|_{[\ell]\setminus U} \mid v|_{U} = u\}$.
By choice of $u$, we have $|S'| \geq L/q^c$.

\noindent By the induction hypothesis, for 
\[ d' \geq \frac{1}{c} \log \Bigl(\frac{L}{2q^{c}}\Bigr) - \Bigl( 1- \frac{1}{c}\Bigr) \log ((q-1)(\ell - c)) \ , \]
there exist $w' \in \F_q^{\ell-c}$ and $v'_1, \ldots, v'_{d'} \in S'$ such
that for each $j \in [d']$,
$$\biggl|\supp(v'_j + w') \setminus \Bigl(\bigcup_{i = 1}^{j-1}\supp(v'_i + w')\Bigr)\biggr| \geq c.$$

Let $d=d'+1$. Note that 
\[ d \ge \frac{1}{c}\log_q\Bigl(
  \frac{L}{2}\Bigr) - \Bigl(1 - \frac{1}{c}\Bigr)\log_q ((q-1)\ell) \ .\]
For $i \in [d']$, let $v_i = (u_0, v'_i) \in \F_q^{\ell}$. Let $v_d$ be any vector in $S$ where $(v_d)|_U$ differs
from $u_0$ in every coordinate of $U$. Let $w = (-u_0, w')$.
We claim that $w$ and $v_1, \ldots, v_d$ satisfy the desired properties. 

\noindent Indeed, for each $j \in [d']$, we have
\begin{align*}
\biggl|\supp(v_j + w) \setminus \Bigl(\bigcup_{i = 1}^{j-1}\supp(v_i + w)\Bigr)\biggr| &=
\biggl|\supp(v'_j + w') \setminus \Bigl(\bigcup_{i = 1}^{j-1}\supp(v'_i + w')\Bigr)\biggr| \geq c.
\end{align*}
\noindent Also, 
$$\biggl|\supp(v_d + w) \setminus \Bigl(\bigcup_{i = 1}^{d-1}\supp(v_i + w)\Bigr)\biggr| \geq \Big| \supp(v_d + w) \setminus ([\ell] \setminus U )\Bigr| = |U| = c.$$
\noindent Thus for all $j \in [d]$, we have 
$$\biggl|\supp(v_j + w) \setminus \Bigl(\bigcup_{i = 1}^{j-1}\supp(v_i + w)\Bigr)\biggr| \geq c,$$
as desired.
\end{proof}


Given Lemma~\ref{lem:qcinc}, the proof of
Theorem~\ref{thm:qcbound} is virtually identical to the proof of its
binary analog Theorem~\ref{thm:cbound}. Theorem~\ref{thm:qmain} can
then be proved (using Theorem~\ref{thm:qcbound}) in the same manner as
Theorem~\ref{thm:main} was proved.

\section*{Acknowledgements}

Some of this work was done when we were all participating in the Dagstuhl seminar 09441
on constraint satisfaction. We thank the
organizers of the seminar for inviting us, and Schloss Dagstuhl for the wonderful hospitality.

\bibliographystyle{alpha}
\bibliography{random-lin-ld}

\end{document}